\documentclass[aip,reprint,onecolumn,citeautoscript]{revtex4-1}

\usepackage{amsmath,amsthm,amssymb}
\usepackage[utf8]{inputenc}
\usepackage[T1]{fontenc}
\usepackage{mathptmx}
\usepackage{etoolbox}
\usepackage{hyperref}
\usepackage{microtype}
\usepackage[capitalise]{cleveref}
\usepackage{physics}

\makeatletter
\def\@email#1#2{%
 \endgroup
 \patchcmd{\titleblock@produce}
  {\frontmatter@RRAPformat}
  {\frontmatter@RRAPformat{\produce@RRAP{*#1\href{mailto:#2}{#2}}}\frontmatter@RRAPformat}
  {}{}
}%
\makeatother

\newcommand{\numberset}{\mathbb} 
\newcommand{\Z}{\numberset{Z}} 
\newcommand{\R}{\numberset{R}}
\let\C\relax
\newcommand{\C}{\numberset{C}}
\newcommand{\inv}{^{-1}}
\let\U\relax
\newcommand{\U}{\mathcal{U}_q}
\newcommand{\qsl}{\mathfrak{sl}_2}

\newcommand{\Hmatrix}{{\hat{H}}}
\newcommand{\mode}{{\tilde{e}}}

\newcommand{\ii}{\mathrm{i}}
\newcommand{\ee}{\mathrm{e}}

\let\oldcite\cite
\renewcommand{\cite}{~\oldcite}

\theoremstyle{plain}
\newtheorem{theorem}{Theorem}[section]
\newtheorem{theoremA}{Theorem}

\newtheorem*{theorem*}{Theorem}
\newtheorem*{corollary*}{Corollary}
\newtheorem*{assumption*}{Assumption}
\newtheorem*{question*}{Question}
\newtheorem{proposition}[theorem]{Proposition}
\newtheorem{lemma}[theorem]{Lemma}
\newtheorem{corollary}[theorem]{Corollary}
\theoremstyle{definition} 
\newtheorem{definition}[theorem]{Definition}

\newtheorem{remark}[theorem]{Remark}

\begin{document}

\title{Hofstadter-Toda spectral duality and quantum groups}

\author{Pasquale Marra} 
\email{pmarra@ms.u-tokyo.ac.jp}
\affiliation{
Graduate School of Mathematical Sciences,
The University of Tokyo, 3-8-1 Komaba, Meguro, Tokyo, 153-8914, Japan}
\affiliation{
Department of Physics \& Research and Education Center for Natural Sciences, 
Keio University, 4-1-1 Hiyoshi, Yokohama, Kanagawa, 223-8521, Japan}
\author{Valerio Proietti}
\affiliation{
Department of Mathematics, University of Oslo, P.O. box 1053, Blindern, 0316 Oslo, Norway}
\email{valeriop@math.uio.no}
\author{Xiaobing Sheng}
\affiliation{
Okinawa Institute of Science and Technology Graduate University, 1919-1 Tancha, Okinawa 904-0495, Japan}
\email{xiaobing.sheng@oist.jp}
\date{\today}

\begin{abstract}
The Hofstadter model allows to describe and understand several phenomena in condensed matter such as the quantum Hall effect, Anderson localization, charge pumping, and flat-bands in quasiperiodic structures, and is a rare example of fractality in the quantum world. An apparently unrelated system, the relativistic Toda lattice, has been extensively studied in the context of complex nonlinear dynamics, and more recently for its connection to supersymmetric Yang-Mills theories and topological string theories on Calabi-Yau manifolds in high-energy physics. Here we discuss a recently discovered spectral relationship between the Hofstadter model and the relativistic Toda lattice which has been later conjectured to be related to the Langlands duality of quantum groups. Moreover, by employing similarity transformations compatible with the quantum group structure, we establish a formula parametrizing the energy spectrum of the Hofstadter model in terms of elementary symmetric polynomials and Chebyshev polynomials. The main tools used are the spectral duality of tridiagonal matrices and the representation theory of the elementary quantum group.
\end{abstract}


\maketitle

\section*{Introduction}\label{sec:intro} 

Among the diversity of natural phenomena, intriguing connections sometimes emerge, where different physical systems share similar mathematical equations or structures. 
These unexpected parallels may seem coincidental or, conversely, hint at a deeper unity within the physical and mathematical worlds. 
One such example is the relationship between the relativistic Toda lattice model\cite{toda_vibration_1967,ruijsenaars_relativistic_1990} in high energy physics and the Hofstadter\cite{harper_single_1955,hofstadter_energy_1976,aubry_analyticity_1980} model in condensed matter physics.

The Toda lattice model describes a one-dimensional lattice of $N$ particles with exponentially decaying interactions\cite{toda_vibration_1967}.
Initially conceived as a simple toy model for one-dimensional crystals, it is a prominent example of a nonlinear system in mathematical physics, which stands out for the remarkable property of being exactly solvable and integrable.
Moreover, the model exhibits the emergence of solitonic excitations, providing insights into the complex dynamics of nonlinear systems.
Furthermore, its relativistic generalization\cite{ruijsenaars_relativistic_1990} is studied in high-energy physics for its connection to supersymmetric Yang-Mills theories in five dimensions\cite{nekrasov_five-dimensional_1998} and topological string theories on Calabi-Yau manifolds\cite{katz_geometric_1997,aganagic_quantum_2012,hatsuda_exact_2016,hatsuda_hofstadters_2016,hatsuda_calabi-yau_2017}.

On the other hand, the Hofstadter model describes noninteracting charged fermions (e.g., electrons) on a two-dimensional lattice in a magnetic field (more generally, a gauge field) perpendicular to the plane or, alternatively, noninteracting fermions in a one-dimensional quasiperiodic lattice\cite{harper_single_1955,hofstadter_energy_1976,aubry_analyticity_1980}.
The most distinctive feature of the model is the presence of a fractal energy spectrum, which is one of the few examples of fractals in quantum physics, and the very rich phase diagram\cite{osadchy_hofstadter_2001} with topological gapped phases indexed by a topological invariant, the Chern number of the occupied bands\cite{thouless_quantized_1982}.
The Hofstadter model is the discrete counterpart of a quantum Hall system, in the sense that it describes the Landau levels of free electrons in a magnetic field regularized on a discrete lattice\cite{thouless_quantized_1982}.
Indeed, in the weak field limit, i.e., when the magnetic cyclotron radius becomes much larger than the lattice constant, the quantum Hall system and the Hofstadter model become physically equivalent.
It is however when the magnetic cyclotron radius is comparable with the lattice constant that the fractal properties emerge, as a consequence of the incommensuration between these two characteristic lengths.

The Hofstadter Hamiltonian is moreover identical to the model introduced by Aubry and André\cite{aubry_analyticity_1980}, describing fermions on a one-dimensional discrete lattice in the presence of a harmonic potential not necessarily commensurate with the lattice parameter.
In the Aubry-André model, the role of the magnetic flux per unit cell in the two-dimensional system is played by the angular wavenumber of the harmonic potential (in units of the lattice parameter), and the role of the second dimension is played by the phase shift of the harmonic potential with respect to the discrete lattice, giving rise to a so-called second \emph{synthetic} dimension. 
This also leads to the emergence of topological charge pumping\cite{thouless_quantization_1983}, which shares the same topological origin as the quantum Hall effect.
The fractal nature of the spectra emerges here as a consequence of the incommensuration between the periodicity of the harmonic potential and the periodicity of the underlying lattice.
This model can describe the effects of disorder, including Anderson localization\cite{aubry_analyticity_1980} and quasiperiodicity\cite{kraus_topological_2012,kraus_four-dimensional_2013,marra_topologically_2020}.

Physical realizations of the Hofstadter model include microwave photonic systems\cite{kuhl_microwave_1998}, Moiré superlattices in graphene\cite{ponomarenko_cloning_2013,dean_hofstadters_2013,hunt_massive_2013}, ultracold atoms in incommensurate optical lattices\cite{aidelsburger_realization_2013}, and interacting phonons in superconducting qubits\cite{roushan_spectroscopic_2017}.

Although originating from distinct branches of physics, the $N=2$ relativistic Toda lattice and Hofstadter models are described by Hamiltonians which transform one into the other by a formal substitution of the position and momenta operators $x,p\mapsto \ii x,\ii p$.
Mathematically, this corresponds to substituting Mathieu ($\cos$) and the modified Mathieu ($\cosh$) potentials or operators.
Recently, a deep relationship between the energy spectra of the Hofstadter model, the relativistic Toda lattice model, and their modular dual was found by Hatsuda, Katsura, and Tachikawa\cite{hatsuda_hofstadters_2016}.
Building on these results, it has been conjectured that the Hofstadter model is intimately related to the Langlands duality of quantum groups\cite{ikeda_hofstadters_2018}, suggesting a deep connection between the Langlands program, quantum geometry, and quasiperiodicity. 

Here, using the representation theory of the elementary quantum group, and the fundamental properties of tridiagonal matrices, we study the properties of a polynomial which appears to capture spectral information for both models. More precisely, the polynomial is implicitly defined by 
the fundamental self-recurring property 
of the Hofstadter butterfly.

We notice that both the Hofstadter and Toda lattice models 
are known to be related to the root systems $A_{N-1}$.
Hence it should be no surprise that \(\U(\qsl)\) (the simplest \emph{quantum group}) plays a central role in the study of these models. 
More generally, we stress the fact that noncommutative geometry is the most natural language to investigate the properties of the Hofstadter model and the related Aubry and André model, as well as other quasiperiodic and aperiodic structures emerging in condensed matter physics, such as quasicrystals, Penrose tilings, and disordered systems\cite{bellissard_the-noncommutative_2003,putnam_non-commutative_2010,loring_disordered_2011,oyono-oyono_c-algebras_2011}.

It should be emphasized that the first result below (Theorem \ref{thm:prel}) has already been understood\cite{hatsuda_hofstadters_2016} in the context of the relativistic Toda model, and our contribution is a clarification of the relationship with the characteristic polynomials associated with the Hofstadter model, making use of 
a ``quantum group-adjusted'' gauge introduced by Wiegmann and Zabrodin\cite{wiegmann_bethe-ansatz_1994},
and further studied by Hatsugai, Kohmoto, and Wu\cite{hatsugai_quantum_1996},
and the spectral duality developed by Molinari (see Ref.~\onlinecite{mol:tri}, where general Hamiltonian matrices similar to the one considered here are investigated).

The formula given in Theorem \ref{thm:formula_conc} and the parametrization in Theorem \ref{thm:par} are new, and they are obtained by using the representation theory of $\U(\qsl)$ and standard formulas involving Chebyshev polynomials. 

Let us briefly summarize our results. 
\begin{theoremA}[Ref.~\onlinecite{hatsuda_hofstadters_2016}]
\label{thm:A}
 Let $E,\tilde{E}$ denote the energy of the Toda Hamiltonian and its modular dual, respectively. The spectral transformation $E\mapsto \tilde{E}$, which also appears 
 as the self-similarity relation of the Hofstadter butterfly, 
 is defined by 
 \begin{equation}\label{eq:pprel}
 \det(E-\Hmatrix_{P/Q})=\det(\tilde{E}-\Hmatrix_{Q/P})
 \end{equation}
 where $\Hmatrix_{P/Q}$ is the finite-dimensional Hofstadter Hamiltonian with rational flux $\alpha=P/Q$, evaluated at the point $\vec{\nu} =(\pi/2Q,\pi/2Q)$ of the Brillouin zone.
\end{theoremA}

The polynomial relationship of Eq. \eqref{eq:pprel} was first proved 
by Hatsuda, Katsura, and Tachikawa\cite{hatsuda_hofstadters_2016} 
by comparing the eigenvalue equation of the Hofstadter Hamiltonian $H$, the Toda Hamiltonian $H_\text{Toda}$, and its modular dual.

Here, we clarify the argument by using the spectral duality of tridiagonal matrices with corners developed in Ref.~\onlinecite{mol:tri} and proving that the polynomials involved are indeed the characteristic polynomials determining the spectrum of the Hofstadter model (as observed in Ref.~\onlinecite[p.~10]{hatsuda_hofstadters_2016}).

\begin{theoremA}\label{thm:B}
We have a formula for $f(E)=\det(\Hmatrix_{P/Q}(\kappa_+,\kappa_-)-E)$ given as
\begin{gather*}
f(E)=\sum_{i=0}^{\lfloor Q/2\rfloor} (-1)^{Q+i}\,4^i\,E^{Q-2i}\,\mode_i(\sin^2(\pi\alpha),\sin^2(2\pi\alpha),\dots,\sin^2((Q-1)\pi\alpha)
\\+ 4\epsilon\cdot\cos(Q\kappa_-)\sin(Q\Bigr(\kappa_+ -\frac{\pi}{2}\Bigl)),
\end{gather*}
where $\epsilon\in\{\pm 1\}$ is a sign and $\vec{\nu}\mapsto \kappa_\pm$ is a transformation of the Brillouin torus.
\end{theoremA}

For more details on the result above, see Theorem \ref{thm:par} and related discussion above it. 
This result should be interpreted as follows: in order to gain more insight into the mapping $E\mapsto \tilde{E}$, we can try to obtain explicit formulas for $f$ for all values of $\alpha=P/Q$ and then invert one side of Eq. \eqref{eq:pprel} (at least locally) to get an analytic expression $\tilde{E}=\tilde{E}(E)$. 
This theorem is an attempt towards this goal.
It should be noted that a formula for the Hofstadter spectrum at the mid-band point appeared previously in Ref.~\onlinecite{kreft_93}, by using a strategy related to the one employed in this paper, but without involving the notion of quantum group.

The layout of the paper is as follows. 
We first briefly review the relevant mathematical structures (generalized Clifford algebras, rotation algebra, the elementary quantum group) and physical models 
from a unified perspective based on the principal series representation of $\U(\qsl)$. 
The second section starts with a brief recall of the spectral duality of tridiagonal matrices and continues on proving Theorem \ref{thm:A} above. 
In the third section, 
we exploit the irreducible representations of the quantum group to turn $\Hmatrix$ into a bidiagonal matrix without corners and with one constant diagonal. This allows us to establish a formula for its characteristic polynomial (Theorem \ref{thm:B}). 
In the fourth section, we employ
Chebyshev polynomials in order to give a different parametrization of $\Hmatrix$ over the Brillouin torus (equivalent to the Chambers relation).
We then finish with a conclusion and a brief outlook.

\section{Preliminaries on mathematical structures and models}\label{sec:prelims}

In this section, we will introduce the Hofstadter and Toda models in a somewhat unified way by emphasizing the key mathematical notions which underlie both systems. 

Let us start by introducing the generalized Clifford algebra of order $n$ on three generators, i.e., elements $\{e_1,e_2,e_3\}$ satisfying
the Weyl braiding relations $e_ie_j=\omega e_je_i$ $(i<j)$ and \(e_i^n=\omega^n=1\) for a primitive $n$-th root of unity $\omega$. As the
classical case ($n=2, \omega=-1$) is generated by the Pauli matrices, the matrix representations of the $e_i$'s are also known as 
generalizations of the Pauli matrices. 

There are several such constructions, most notably the Gell-Mann matrices and Sylvester's
shift and clock matrices\cite{sylv:three}. The former are Hermitian and traceless (just like the Pauli matrices), while the latter are unitary
and traceless, which makes them preferable for our goals. 

\begin{definition}
 
Set \(e_1=V, e_2=VU, e_3=U\). Sylvester's construction is given as
\begin{align*}
 V &= \begin{pmatrix}
 0 & 1 & 0 & \cdots & 0\\
 0 & 0 & 1 & \cdots & 0\\
 0 & 0 & \ddots & \ddots & \vdots\\
 \vdots & \vdots & \vdots & \ddots & 1\\
 1 & 0 & 0 & \cdots & 0
 \end{pmatrix}, &
 U &= \begin{pmatrix}
 1 & 0 & 0 & \cdots & 0\\
 0 & \omega & 0 & \cdots & 0\\
 0 & 0 & \omega^2 & \cdots & 0\\
 \vdots & \vdots & \vdots & \ddots & \vdots\\
 0 & 0 & 0 & \cdots & \omega^{n-1}
 \end{pmatrix}.
\end{align*}
These are called the \emph{shift} and \emph{clock} matrices, respectively.
\end{definition}

The connection with quantum mechanics is evident from the explicit form of $V$ and $U$: the clock matrix $U$ amounts to the exponential
of position in a periodic space of $n$ ``hours'' 
(e.g., discrete lattice sites), 
and the circular shift matrix $V$ is just the translation operator in this cyclic
space, i.e., the exponential of the momentum.

Weyl's formulation of the canonical commutation relations reads \begin{equation}\label{eq:weylbraid}
UV(VU)^{-1}=UVU^{n-1}V^{n-1}=\omega\inv,
\end{equation} 
which justifies the name \emph{Weyl-Heisenberg matrices} for $V$ and $U$. 
The Hofstadter model is described by the Hamiltonian $H= V+V^*+U+U^*$.
Before going further, however, we need to introduce some additional mathematical structures. 
The relation in Eq. \eqref{eq:weylbraid} defines a universal $C^*$-algebra \(A_\omega\) studied by Rieffel
\cite{rieffel:nc}, known as noncommutative torus or \emph{rotation algebra}. It is arguable that the naming ``noncommutative torus'' should be reserved for the case \(\omega=\exp(2\pi \ii \alpha)\) and \(\alpha\) is \emph{irrational}, unlike what we assumed so far. This is because
when \(\alpha\) is \emph{rational}, \(A_\omega\) is Morita equivalent to the \emph{commutative} algebra of continuous functions on the
standard torus. By contrast, when $\alpha$ is irrational, the Morita equivalence class is determined by the orbit of the modular action of
\(\mathrm{PGL}_2(\Z)\) by linear fractional transformations on $\alpha$. Being these generated by \begin{equation}\label{eq:modtrans}
z\mapsto z+1,\qquad z\mapsto z\inv, \end{equation} in the mathematical physics literature $A_\omega$ is said to be in \emph{modular duality}
with $A_{\tilde{\omega}}$, where \(\tilde{\omega}=\exp(2\pi \ii \alpha\inv)\). Let us remark here that the operator algebras community does
not use this naming convention. 
However, other related dualities, such as Spanier-Whitehead duality in $K$-theory\cite{connes:ncg,val:shi} or the property of being each other's von Neumann algebra commutant (with coupling constant $\alpha$)\cite{fad:weyl}, are more frequently mentioned for
\(A_\omega\) and \( A_{\tilde{\omega}} \).

\Cref{eq:modtrans} provides our second contact point with the Hofstadter model, as those two transformations also appear as the
fundamental 
self-similarity relations of the Hofstadter butterfly\cite{hofstadter_energy_1976,hatsuda_hofstadters_2016,ikeda_hofstadters_2018}.
In addition, the modular duality discussed above has inspired Faddeev's notion of \emph{modular double} of a quantum group\cite{fad:moddouble}. The modular double of $\U(\qsl)$ plays a central role in the study of the relativistic Toda lattice\cite{kharchev_unitary_2002}, particularly in the argument of Ref.~\onlinecite{hatsuda_hofstadters_2016} where the polynomial relation which inspired this paper is established. To understand this, let us first recall the definition of the elementary quantum group. 

\begin{definition}
 The quantum group \(\U(\qsl)\) is the 
associative 
 algebra generated by elements \(E, F, K^{\pm 1} \) satisfying the relations
\begin{gather*} KE=q^2 EK,\qquad KF=q^{-2}FK\\ EF-FE=\frac{K-K\inv}{q-q\inv}, 
\end{gather*} 
where $q=\omega^{1/2}=\ee^{\ii\pi\alpha}$.
\end{definition}

In this paper, we will not need the coproduct which makes \(\U(\qsl)\) a bialgebra. When $\alpha$ is real, the commutation relations above are compatible with a certain involution, giving a real form for $\U(\qsl)$. 
The key
point is that the principal series representations of this quantum group factors through the rotation algebra $A_{q^2}$. Indeed, the following
proposition is readily verified. 

\begin{proposition}\label{prop:rept}
For any \(t\in \C \), there is a morphism of associative algebras $\phi\colon \U(\qsl)\to A_{q^2}$ given as follows:
\begin{equation*}
\phi(K)=tv\inv,\qquad \phi(E)=\frac{u\inv(1-v\inv)}{q-q\inv},\qquad \phi(F)=\frac{qu(t-t\inv v)}{q-q\inv}.
\end{equation*}
\end{proposition}
Above, we use lowercase letters to distinguish $A_\omega$'s generators from Sylvester's matrices.
By virtue of modular duality, setting $\alpha=P/Q$ and $\tilde{q}=\tilde{\omega}^{1/2}$, we have the following \emph{commuting}
representations of $A_{q^2}$ and $A_{\tilde{q}^2}$ on $L^2(\R)$:
\begin{gather}\label{eq:reptorus} \rho_0(v)=T_P,\quad \rho_0(u)=S_{-\ii Q},
\qquad \tilde{\rho}_0(\tilde{v})=T_Q,\quad \tilde{\rho}_0(\tilde{u})=S_{-\ii P},\\\notag T_s f(x)=f(x+s),\qquad S_{-is}f(x)=\exp(\frac{2\pi \ii x}{s})f(x).
\end{gather}
Following Ref.~\onlinecite{kharchev_unitary_2002}, we operate a Wick rotation on \(\rho_0, \tilde{\rho}_0 \), by changing $T_s$ into \(T_{\ii s}\) and \(S_{-\ii s} \) into \(S_{s}\). Note there is a fundamental difference between the representations of the quantum group
and those of the rotation algebra: the former are obtained from \emph{unbounded} representations of the latter.
The newly obtained representations will be denoted \(\rho, \tilde{\rho} \). 
Since they still commute, we can combine them with $\phi$ to form a representation of the modular double quantum group \begin{equation*} \U
(\qsl)\otimes \mathcal{U}_{\tilde{q}}(\qsl). \end{equation*}

The general definition of this object involves the Langlands dual Lie algebra\cite{kharchev_unitary_2002}. 
However, we do not need to worry about such details here as $\qsl$ is self-dual. 
We can anticipate that the Toda Hamiltonian is 
$H\propto\rho\phi(E+F)$.

In condensed matter physics, the $N$-particle Toda lattice model\cite{toda_vibration_1967} is a seminal toy model describing massive particles on a one-dimensional lattice subject to a nearest-neighbor and exponentially decaying two-body interaction.
In the periodic case (particles on a circle) the Hamiltonian reads 
\begin{equation}\label{eq:nrToda}
H=\sum_{n=1}^N \frac12 p_n^2 + (\ee^{-(q_n-q_{n-1})}+\ee^{-(q_{n+1}-q_{n})})
\end{equation}
with the periodic boundary condition $q_0\equiv q_N$, $q_{N+1}\equiv q_1$, where $p_n$ are the canonical momenta and $q_n$ the canonical coordinates satisfying the canonical Poisson bracket relations $\{q_n,p_m\}=\delta_{nm}$.
This Hamiltonian is nonrelativistic, being invariant up to Galilean transformations but not transformations of the Poincaré group, and is an early example of the family of integrable one-dimensional many-body problems, also called Calogero-Moser systems.
A general approach to obtain a relativistic generalization of the Toda lattice model (as well as other Calogero-Moser systems) is to deform the nonrelativistic Hamiltonian to make it compatible with the Poincaré group (see Ref.~\onlinecite{ruijsenaars_relativistic_1990} for more details).
This involves maintaining the nearest-neighbor exponential form of the interaction term, exponentiating the canonical momenta, and combining the resulting terms into an Hamiltonian which will coincide with the time-translation generator of the Poincaré group.
Moreover, we want our Hamiltonian to be quantum, that is, to elevate the classical variables to quantum operators satisfying the canonical commutation relation $[q_n,p_m]=\ii\hbar\delta_{nm}$. 
Following this general recipe and using the same approach as Refs.~\onlinecite{kharchev_unitary_2002,hatsuda_exact_2016}, we define our $N$-particle relativistic Toda lattice model as
\begin{equation*}
H_{\text{Toda}}(N)=
\sum_{n=1}^N 
\left(
1+
\lambda^2\ee^{q_{n}-q_{n+1}}
\right)\ee^{\lambda p_n},
\end{equation*}
with 
$\lambda\in\mathbb{R}$, which reduces to \cref{eq:nrToda} at the second order in $\lambda$.
In this sense, the relativistic (and quantum) Toda lattice model is a one-parameter deformation of the classical Toda lattice model.
For $N=2$ particles in the center of mass frame $p_1=-p_2$ and periodic boundary conditions $q_3=q_1$, setting $p= p_1$, $x=(q_1-q_2)$, and taking $\lambda=1$, the Hamiltonian becomes
\begin{equation*}
H_{\text{Toda}}=
\left(
1+
\ee^{x}
\right)\ee^p
+
\left(
1+
\ee^{-x}
\right)\ee^{-p}
,
\end{equation*}
with $[x,p]=\ii\hbar$.
To obtain a notation more suited for our scope, we rescale the coordinate $x\to2\pi\alpha x$ which yields 
\begin{equation}\label{eq:maybetoda}
H_{\text{Toda}}=
\left(
1+
\ee^{2\pi\alpha x}
\right)\ee^p 
+
\left(
1+
\ee^{-2\pi\alpha x}
\right)\ee^{-p} 
,
\end{equation}
and also gives $[x,p]=\ii\hbar/(2\pi\alpha)$.
The relationship between quantum mechanics and quantum geometry is thus revealed if one sets $\hbar=2\pi\alpha$, which leads us back to $[x,p]=\ii$ and $q=\ee^{\ii\pi\alpha}$.
This Hamiltonian hence is expressed in terms of $\rho(v)$ and $\rho(vu)$, namely the first two generators of the generalized Clifford algebra introduced above (with the difference that these are infinite-dimensional operators).
Finally, to uniform our notations to that of Ref.~\onlinecite{hatsuda_hofstadters_2016}, we introduce the canonical transformation $2\pi\alpha x'=2\pi\alpha x+p$ (preserving the commutation relation $[x',p]=[x,p]=\ii$).
By doing so and neglecting the prime, this yields the usual form of the \emph{relativistic Toda Hamiltonian} 
for $N=2$ particles: 
\begin{equation}\label{eq:todaham}
 H_{\text{Toda}}=
\ee^{p}+\ee^{2\pi\alpha x}+\ee^{-p}+\ee^{-2\pi\alpha x}, 
\end{equation}
with standard commutation relation $[x,p]=\ii$. 

\begin{remark}
The Hamiltonians in \cref{eq:maybetoda,eq:todaham} are equivalent up to a change of the canonical coordinates.
It is interesting also to
 note that one Hamiltonian can be 
 obtained from the other under the relabeling of Clifford
 generators $e_2\leftrightarrow e_3$
 which we used above. 
\end{remark}

With a slight abuse of language, even when $P$ and $Q$ are positive integer 
coprimes
such that $P/Q$ is a fraction in reduced form, we define the
modular dual Hamiltonian $\tilde{H}_{\text{Toda}}$ by swapping $P$ and $Q$ in Eq. \eqref{eq:todaham}. Let us note here that
\(H_{\text{Toda}}\) is obtained from \((q-q\inv)\rho\phi(E+F)\) by a suitable relabeling of Clifford generators:
\begin{align*}
 (q-q\inv)\rho\phi(E+F) &= u\inv - (vu)\inv + u - q^2 uv && (t=q\inv)\\
 & = u\inv - (vu)\inv + u - vu && (uv=q^{-2}vu)\\
 & = u\inv + v\inv + u + v && (e_2\leftrightarrow -e_1).
\end{align*}
This motivates us to find a similar expression for the Hofstadter model. 

The \emph{Hofstadter Hamiltonian} is related to the $N=2$ relativistic Toda Hamiltonian defined in Eq. \eqref{eq:todaham} by Wick rotations $p\to\ii p$, $x\to\ii x$ of both momentum and position operators
\begin{equation}\label{eq:hofham}
 H_{\text{Hof}}=
\ee^{\ii p}+\ee^{2\pi \ii\alpha x}+\ee^{-\ii p}+\ee^{-2\pi \ii\alpha x}. 
\end{equation}
Of course, this is not the conventional way to introduce $H_{\text{Hof}}$, so let us briefly outline a more orthodox derivation. 

The Hofstadter Hamiltonian describes the Hamiltonian of a charged particle in a magnetic field and in a periodic potential $V(x,y)$ describing a two-dimensional lattice.
There are two opposite ways to derive this Hamiltonian:
In the limit of strong fields, one 
takes the Hamiltonian of a charged particle in a field (describing the quantum Hall effect), which is diagonalized in terms of Landau levels, 
and then treats the periodic potential $V(x,y)$ as a perturbation on each Landau level (as done by Thouless-Kohmoto-Nightingale-den Nijs\cite{thouless_quantized_1982}).
On the other hand, in the limit of weak fields, one starts with the so-called ``tight-binding'' Hamiltonian describing the lowest energy sector of a charged particle in a periodic potential, and treats the magnetic field as a perturbation (as done by Harper\cite{harper_single_1955} and Hofstadter\cite{hofstadter_energy_1976}).
These two approaches lead to the same Hamiltonian, with one remarkable exception, as already noted in Ref.~\onlinecite{thouless_quantized_1982}:
the value of the magnetic flux is replaced as $\Phi\to1/\Phi$.
As noted in Ref.~\onlinecite{hatsuda_hofstadters_2016}, this formal operation coincides precisely with taking the modular dual of the Hamiltonian.

In the case of weak fields, following Ref.~\onlinecite{kohmoto_zero_1989,hatsuda_hofstadters_2016} one writes the tight-binding Hamiltonian describing a charged particle in a two-dimensional periodic potential in the usual form $H= T_x+T_x^*+T_y+T_y^*$ where $T_x=\sum_{nm} c^*_{n+1,m}c_{n,m}$ and $T_y=\sum_{nm} c^*_{n,m+1}c_{n,m}$ are the translation operators on the lattice along the two axes $x$ and $y$, with $c_{n,m}$ the fermion operators on the lattice at on the position $(x,y)=(n,m)$ with $n,m\in\Z$.
We now recall that the translation operator corresponding to a displacement $\vec{r}$ is also $T(\vec{r})=\ee^{-\ii \vec{r}\cdot \vec{p}}$, which gives $T_x=\ee^{-\ii p_x}$, $T_y=\ee^{-\ii p_y}$ (we are translating by one lattice site).
This gives
\begin{align}
H&= T_x+T_x^*+T_y+T_y^*=
\ee^{\ii p_x}+\ee^{-\ii p_x}+\ee^{\ii p_y}+\ee^{-\ii p_y}\\
&=\sum_{nm} c^*_{n+1,m}c_{n,m} + c^*_{n,m+1}c_{n,m} +\text{h.c.},
\end{align}
which also immediately gives the energy dispersion $E(k_x,k_y)=2\cos{k_x}+2\cos{k_y}$, i.e., the energy parameterized as a function of the momentum eigenvalues $k_x,k_y$.
The Hamiltonian is not gauge invariant, but it becomes gauge-invariant by performing the Peierls substitution $c^*_{n,m}c_{n',m'}\to\ee^{\ii \phi}c^*_{n,m}c_{n',m'}$ where the gauge-invariant phase $\phi$ is defined as $\phi=\int \vec{A} \cdot \dd \vec{r}$ where $\vec{r}=(x,y)$, $\vec{A}=(A_x,A_y)$ is the vector potential, and where the integral is evaluated on the shortest path connecting the lattice sites $(n',m')$ and $(n,m)$\cite{kohmoto_zero_1989}.
This substitution accounts for the presence of the magnetic field (as a weak perturbation), and results in promoting the translation operators to magnetic translation operators $T_x=\sum_{nm} c^*_{n+1,m}c_{n,m}\ee^{\ii \phi_x}$ and $T_y=\sum_{nm} c^*_{n,m+1}c_{nm} \ee^{\ii \phi_y}$ with the phases $\phi_x=\int_n^{n+1} A_x(x,y)\dd x$, $\phi_y=\int_m^{m+1} A_y(x,y)\dd y$.
Now, by taking a constant field perpendicular to the lattice $B_z=2\pi\alpha$ and choosing the Landau gauge $\vec{A}=(0,2\pi\alpha x)$ so that $\phi_x=0$ and $\phi_y=2\pi\alpha x$, one obtains the Hamiltonian
\begin{align}
H&= T_x+T_x^*+T_y+T_y^*=
\ee^{-\ii p_x}+\ee^{\ii p_x}+\ee^{-\ii (p_y- 2\pi\alpha x)}+\ee^{\ii (p_y- 2\pi\alpha x)}
\nonumber\\&=
\sum_{nm} c^*_{n+1,m}c_{n,m} + c^*_{n,m+1}c_{n,m} \ee^{-2\pi \ii n\alpha} +\text{h.c.}
\end{align}
The same result can be obtained directly by performing the substitution $\vec{k}\mapsto\vec{p}-\vec{A}$ directly in $E(k_x,k_y)$, promoting the energy dispersion to the Hamiltonian operator $H=2\cos{p_x}+2\cos{(p_y-2\pi\alpha x)}$\cite{hofstadter_energy_1976,kohmoto_zero_1989}.
In the limit of strong fields instead, by treating the periodic potential $V(x,y)$ as a perturbation on each Landau level, one obtains the same Hamiltonian but with the remarkable exception that one now gets $\alpha\to1/\alpha$, as already mentioned (see Ref.~\onlinecite{thouless_quantized_1982}).
Since $[p_y,x]=[p_y,p_x]=0$, one has that $[p_y,H]=0$:
Hence it is possible to simultaneously diagonalize the Hamiltonian and the momentum $p_y$.
In particular, for eigenvalues of the momentum $p_y\to0$, the Hamiltonian reduces to the Hamiltonian $H_{\text{Hof}}$ defined in \cref{eq:hofham}.

One can introduce an anisotropy parameter $R$ in the Hamiltonian, by taking $H= T_x+T_x^*+R(T_y+T_y^*)$.
The parameter $R$ can be understood as a measure of the degree of anisotropy of the system, 
with $R=1$ corresponding to $x$ and $y$ directions being perfectly symmetric.
To simplify notations, we will assume $R=1$ hereafter, except when explicitly noted.

Since $[p_y,H]=0$, we can assume wavefunctions as a product of plane waves in the $y$-direction and the usual Bloch functions in the $x$ direction, 
leading us to the following \emph{Harper equation}:
\begin{gather}
 \psi(n,m)= \ee^{\ii(\nu_x n+\nu_y m)} g_n,\notag\\
 \label{eq:mathieu}
 \ee^{\ii\nu_x}g_{n+1}+\ee^{-\ii\nu_x}g_{n-1}+2 R \cos(2\pi n\alpha-\nu_y)g_n=E g_n,
\end{gather}
where we used again the fact that the exponential of the momentum corresponds to the translation operator. Setting $\nu_x=0$, the left-hand side above
is an operator of $\ell^2(\Z)$ known as the \emph{almost Mathieu operator}, featured in Barry Simon's fifteen problems about Schrödinger
operators ``for the twenty-first century''\cite{barry:century}. This operator comes with a parameter $R>0$ which multiplies the cosine function in Eq. \eqref{eq:mathieu}.

Let us point out that the measure-theoretic properties
of the spectrum of the almost Mathieu operator crucially depend on $R$. 
When $\alpha$ is irrational, the spectrum is a Cantor set of Lebesgue measure $\lvert 4-4R\rvert$.
Hence, the case $R=1$ has spectrum of measure zero, and it can be shown to be 
surely purely singular continuous spectrum. 
The case $R>1$ is notable for having
almost surely pure point spectrum and exhibiting \emph{Anderson localization}
(see Refs.~\onlinecite{aubry_analyticity_1980,avila_solving_2006,avila_the-absolutely_2008}).

Let us assume $\alpha$ is rational now. Imposing $Q$-periodicity, and treating $g_n$ as
the $n$-th coordinate of $g$ in $\C^Q$, the equation above becomes the eigenvalue problem for the matrix below, which we record as a separate definition as it will play a major role in the 
rest of the paper.

\begin{definition}\label{def:hoffindim}
 The \emph{finite-dimensional Hofstadter Hamiltonian}:
\begin{align}
\Hmatrix & = \ee^{\ii\nu_x}V+\ee^{-\ii\nu_x}V^*+\ee^{\ii\nu_y}U+\ee^{-\ii\nu_x}U^*\notag\\
\label{eq:Hofmatrix}
 \Hmatrix & = \begin{pmatrix}
 \begin{smallmatrix}
 2\cos(2\pi\cdot 0 \cdot\alpha-\nu_y)
 \end{smallmatrix} & z & & & z\inv\\
 z\inv & \ddots & z & & \\
 & \ddots & \ddots & \ddots & \\
 & & z\inv & \ddots & z\\
 z & & & z\inv & 
 \begin{smallmatrix}
 2\cos(2\pi(Q-1)\alpha-\nu_y)
 \end{smallmatrix}
 \end{pmatrix},
\end{align}
where we have set the shorthand $z=\ee^{\ii\nu_x}$ (recall that $q^2=\omega$).
\end{definition}

It will be convenient to set $T_x=\ee^{\ii\nu_x}V, T_y=\ee^{\ii\nu_y}U$. 
The group spanned by $T_x,T_y$ is referred to as the group of \emph{magnetic translations}\cite{zak_magnetic_1964}.

Our goal is now to find an expression for $\Hmatrix$ in terms of quantum group elements $E,F$, in analogy with the Toda model,
in the spirit of Refs.~\onlinecite{wiegmann_bethe-ansatz_1994,hatsugai_quantum_1996}.
Toward this goal, we define a finite-dimensional representation $\varphi$ of $\U(\qsl)$ in terms of shift and clock operators. 
The notation suggests that
this should be a variant of the representation $\phi$ introduced above in the Toda model setting. Choosing $t=q^{-1}$ in Proposition \ref{prop:rept} and relabeling the Clifford
generators as above, we define $\varphi$ as follows:

\begin{equation}\label{eq:screp}
\varphi(K)=- q\inv UV\inv ,\qquad \varphi(E)=\frac{V\inv+U\inv}{q-q\inv},\qquad \varphi(F)=\frac{V+U}{q-q\inv}.
\end{equation}
We can now write the expression $\Hmatrix=(q-q\inv)\varphi(E+F)$ which we sought.

\section{Characteristic polynomial and spectral duality}

As already mentioned, 
Hatsuda, Katsura, and Tachikawa\cite{hatsuda_hofstadters_2016} found a relationship between 
the characteristic polynomial of the Hofstadter Hamiltonian $H$, the self-similarity relation defining the fractal properties of the Hofstadter butterfly,
and the spectral transformation between the Toda Hamiltonian $H_\text{Toda}$ and its modular dual.

Let us denote the polynomial induced by the modular duality between the Toda Hamiltonian $H_\text{Toda}$ and its modular dual defined in Ref.~\onlinecite{hatsuda_hofstadters_2016}
by $f=f_{P/Q}(E)$, where $E$ is the energy. 
We are going to prove that $f$ is the characteristic polynomial
$f(E)=\det (\Hmatrix - E)$ of the finite-dimensional Hofstadter Hamiltonian (note that $\Hmatrix$ depends on $P/Q$, although we suppressed this
dependence from the notation).

After that, we will establish a general formula for $f$ in terms of 
a set of polynomials closely related to the
elementary symmetric polynomials. 
Since the fractal properties of the Hofstadter butterfly are encoded by the self-similarity relations
\begin{equation}\label{eq:fundsym}
 (\alpha,E)\mapsto (\alpha+1, E),\qquad (\alpha,E)\mapsto (\alpha\inv, \tilde{E}),
\end{equation}
we hope that our formula will 
lead to a better understanding of the function $E\mapsto \tilde{E}$,
which is unknown so far, and whose physical significance is still unclear, to the best of our knowledge.

Let us briefly introduce some elements of spectral duality for tridiagonal matrices (see Ref.~\onlinecite{mol:tri} for more details).
The Harper equation can be written in matrix form:
\begin{equation}\label{eq:harpiter}
 \begin{pmatrix}
 g_{k+1}\\
 g_{k} 
 \end{pmatrix} =A_1
 \begin{pmatrix}
 E-\theta_k & -1\\
 1 & 0 
 \end{pmatrix} z\inv A_1\inv \begin{pmatrix}
 g_{k}\\
 g_{k-1} 
 \end{pmatrix},
\end{equation}
where we have set $\theta_k=2\cos (2\pi (k-1)\alpha-\nu_y)$ and $A_1=\mathrm{diag}(z\inv,1)$.
Denoting by $T_k(E)$ the matrix in Eq. \eqref{eq:harpiter} and iterating we obtain:
\begin{equation*}
 z^Q 
 \begin{pmatrix}
 g_{1}\\
 g_{Q} 
 \end{pmatrix}
 = T_Q(E)\cdots T_1(E)
 \begin{pmatrix}
 g_{1}\\
 g_{Q} 
 \end{pmatrix}.
\end{equation*}
Let us define the \emph{transfer matrix} $T=T(E)$ as the product of the $T_k$'s appearing above. We see
that $z^Q$ is an eigenvalue of $T(E)$ if and only if $E$ is an eigenvalue of $\Hmatrix=\Hmatrix(z)$. This
exchange of roles between the parameter $z$ and the energy $E$ is the basis of spectral duality. 
\begin{theorem}[Ref.~\onlinecite{mol:tri}]\label{thm:spd}
 We have the equality $f(E,z)=(-1)^{Q-1}z^{-Q}\det(T(E)-z^Q)$.
\end{theorem}
\begin{proof}
 Recall $f(E,z)=\det(\Hmatrix(z)-E)$. By spectral duality this last quantity is zero if and only if
 $\det(T(E)-z^Q)$ is zero. This is the determinant of a $2$-by-$2$ matrix, hence it is $z^{2Q}
 -\trace(T(E))z^Q+\det(T)$. Notice $\det(T)=1$ by multiplicativity. We can match $f$ and
 $\det(T(E)-z^Q)$ as polynomials in $E$ multiplying by $(-1)^{Q-1}z^{-Q}$. 
\end{proof}

Since $\det(T(E)-z^Q)=z^{2Q}-\trace(T(E))z^Q+1$, the previous theorem implies
\begin{equation}\label{eq:dispx}
 f(E,\nu_x)= (-1)^{Q}\trace (T(E))+2(-1)^{Q-1}\cos(Q\nu_x).
\end{equation}

When $\alpha=P/Q$ is a rational number, the energy spectrum of the Hofstadter
model splits into $Q$ energy bands. The eigenvalues of $\Hmatrix$ give us a point in each band, and the
rest of the band is obtained by parametrizing over a torus, 
which coincide with the \emph{Brillouin zone}, whose
coordinates are given by the vector $\vec{\nu}=(\nu_x,\nu_y)$. The way $f$ depends on $\vec{\nu}$ is
well-understood\cite{chambers_linear-network_1965,kohmoto_zero_1989}: it is a simple translation as  shown by the \emph{Chambers relation},
\begin{equation}\label{eq:chambers}
 f(E,\nu_x,\nu_y)=f\Bigr( E,\frac{\pi}{2Q},\frac{\pi}{2Q}\Bigl)+2(-1)^{Q-1}(\cos(Q\nu_x)+\cos(Q\nu_y)).
\end{equation}

Note that Eq. \eqref{eq:dispx} is in fact a proof of the Chambers relation for the $x$-coordinate. We will
see how to obtain the $y$-dependence later after using 
the representation theory of the quantum group. 
We will refer to the points
$\vec{\nu}=(0,0)$, 
$\vec{\nu}=(\textstyle\frac{\pi}{2Q},\frac{\pi}{2Q})$, and
$\vec{\nu}=(\textstyle\frac{\pi}{Q},\frac{\pi}{Q})$ 
as the center, mid-band, and corner points, and restrict ourselves to the reduced Brillouin zone $\nu_x,\nu_y\in[0,2\pi/Q]$, except when explicitly noted.
This can be done without loss of generality, since the Hamiltonian is periodic in $\nu_x,\nu_y$ with a period $2\pi/Q$ up to unitary transformations\cite{marra_fractional_2015}.

It is well-known that the characteristic polynomial in $f(E,\nu_x,\nu_y)$ (of order $Q$) contains only terms with powers $E^{Q-2i}$.
This is because odd powers in $E$ necessarily contains terms in $\cos{(2\pi n\alpha-\nu_y)}$, which cancel each other when summed over $n=0,Q-1$ at the  mid-band point $\nu_y=2\pi/Q$.
Consequently, $f(E,\nu_x,\nu_y)$ is an even function of $E$ and the spectrum is symmetric under $E\mapsto-E$ if $Q$ is even, while $f(E,\nu_x,\nu_y)$ is an odd function and the spectrum at the mid-band point contains $E=0$ if $Q$ is odd, since $f( 0,{\pi}/{2Q},{\pi}/{2Q})=0$.
It is also well-known that $f( 0,{\pi}/{2Q},{\pi}/{2Q})=4(-1)^{Q/2}$ if $Q$ is even, which mandates that $f(0,0,0)=0$ if $Q/2$ is even and $f(0,\pi/Q,\pi/Q)=0$ if $Q/2$ is odd.
Consequently, the spectrum contains $E=0$ at the center point if $Q$ is doubly even, at the corner point if $Q$ is singly even, and at the mid-band point if $Q$ is odd, respectively.
In particular, the zeros $E=0$ in the spectra when $Q$ is even are doubly degenerate and form Dirac cones with a linear dependence on the momentum\cite{wen_winding_1989}.

Let us consider the Hamiltonian $H_{\text{Toda}}$ and its modular dual. 
Since these commute, the authors in Ref.~\onlinecite{hatsuda_hofstadters_2016} argue that, by simultaneous diagonalization, we arrive at a pair of difference equations which can be iterated, revealing a tracial relationship of the form
\begin{equation}\label{eq:trrel}
 \tr \prod_{j=0}^{Q-1}
 \begin{pmatrix}
 T_{j} & -1\\
 1 & 0 
 \end{pmatrix}
 =\tr \prod_{j=0}^{P-1}
 \begin{pmatrix}
 \tilde{T}_{j} & -1\\
 1 & 0 
 \end{pmatrix}
\end{equation}
(the highest index appears leftmost in the product). Above $T_j,\tilde{T}_j $ are defined based on the functions $T(x)=E-2\cosh(x)$, $\tilde{T}(x)=\tilde{E}-2\cosh (\alpha\inv x)$, and by shifting the coordinate as $x\mapsto \ii\cdot 2\pi j\alpha$, $x\mapsto \ii\cdot 2\pi j$. Note that the shift is happening along the \emph{imaginary} axis, effectively rotating the system (inverting the Wick rotation).

\Cref{eq:trrel} should be understood as a relationship between the transfer matrices of the Hofstadter model for flux values of $\alpha$ and $\alpha\inv$. For the simple case $(P,Q)=(2,3)$, Eq. \eqref{eq:trrel} reads as follows:
\[
 E^3-6E-2\cosh(3x) = \tilde{E}^2-4-2\cosh(3x).
\]
Clearly, if the $x$-dependent parts are equal, then evaluating at $x=0$ yields a polynomial relation for the $x$-independent parts $-f_{2/3}(E)-2=f_{3/2}(\tilde{E})-2$.

We are now able to apply the spectral duality discussed above (coupled with the Chambers relation) and prove the first main result of this section.

\begin{theorem}[Ref.~\onlinecite{hatsuda_hofstadters_2016}]\label{thm:prel}
 The spectral transformation $E\mapsto \tilde{E}$ induced by the modular duality is defined by the equation $(-1)^{Q} f_{P/Q}(E) = (-1)^{P} f_{Q/P}(\tilde{E})$ at the mid-band point $\vec{\nu}=(\pi/2Q,\pi/2Q)$.
\end{theorem}
\begin{proof}
 Since $T_j$ and $\tilde{T}_j$ are defined by shifting along the imaginary axis, at $x=0$ the hyperbolic cosine function becomes a simple cosine function. Then the left-hand side of Eq. \eqref{eq:trrel}, evaluated at $x=0$, ought to be the trace of the transfer matrix from Eq. \eqref{eq:harpiter}
 (there is a slight difference due to the convention in 
 Ref.~\onlinecite{hatsuda_hofstadters_2016} that $[x,p]=2\pi \ii \alpha$, but after renormalizing accordingly, 
 the identification with the matrix $T(E)$ in Eq. \eqref{eq:harpiter} holds). 
Then by Theorem \ref{thm:spd} and Eq. \eqref{eq:dispx}, the $x$-dependent part in Eq. \eqref{eq:trrel} corresponds to the $\nu_y$-term in the Chambers relation. 
 
 In particular, since the right-hand side of Eq. \eqref{eq:trrel} is defined via modular duality (swapping $P$ with $Q$), the correction in the definition of $\tilde{T}$ is such that the $x$-dependent terms are always equal for any $P$ and $Q$. Since the polynomial (denoted $P$ in Ref.~\onlinecite{hatsuda_hofstadters_2016}) is defined by effectively setting to zero the $x$-dependent part, we have the identification $P_\alpha(E)=(-1)^Q f_{P/Q}(E,\pi/2Q,\pi/2Q)$.
\end{proof}
\begin{remark}
In Ref.~\onlinecite{hatsuda_hofstadters_2016} the authors define a polynomial $P$ including 
the anisotropy parameter
 $R$, in which case the relevant polynomial is $P_\alpha(E)+2R^{Q}$ (note we are adding $2R^{Q}$ rather than $2$).
\end{remark}

\section{Formula for the characteristic polynomial}

We now work towards finding a formula for $f(E)=\det(\Hmatrix-E)$. We can exploit the gauge invariance of $H_{\text{Hof}}$ and choose a different form for
the vector potential. Following Ref.~\onlinecite{wiegmann_bethe-ansatz_1994}, we can choose 
$\vec{A}= \alpha/2\cdot (-x-y,\, x+y+1)$. 
In this gauge, the mid-band point corresponds to
$\vec{\nu}=(\pi/2,\pi/2)$. Taking this into account, we can rewrite the representation $\varphi$ from Eq. \eqref{eq:screp} in terms of magnetic
translations:
\begin{equation}\label{eq:magrep}
\varphi(K)=- q\inv T_yT_x\inv ,\qquad \varphi(E)=\frac{-(T_x\inv+T_y\inv)}{\ii(q-q\inv)},\qquad \varphi(F)=\frac{T_x+T_y}{\ii(q-q\inv)}.
\end{equation}
As a first step we will represent $\U(\qsl)$ in such a way that $\varphi(F\pm E)$ is Hermitian, thus the (new) candidate
Hamiltonian $\Hmatrix^\prime= \ii(q-q\inv)\varphi(F\pm E)$ is also Hermitian.

As a preliminary step, we are interested in the representation $\varphi=\zeta_0$ described in Ref.~\onlinecite{wiegmann_bethe-ansatz_1994}, given as
\begin{equation}\label{eq:artrep}
 \zeta_0(E)_{k+1,k}=\pm\frac{q^{k}-q^{-k}}{q-q\inv},\qquad \zeta_0(F)_{k,k+1}=\frac{q^{k}-q^{-k}}{q-q\inv},
\end{equation}
where the matrix entries range in $k=1,\dots,Q-1$ (the element $K$ is diagonal, but we do not need this), the upper sign is for odd $P$, while the lower sign is for even $P$. Note $\zeta_0(F\pm E)$ is Hermitian, however we will need to slightly modify $\zeta_0$ to be more suitable for our needs.

From the representation theory of $\U(\qsl)$, when $q^2$ is a
primitive root of unity, we have essentially three families of representations. The first
two are low-dimensional, and hence they must be discarded in favor of the last one, which we denote $Z_{a,b}(\lambda)$. This family is topologically parametrized by a $3$-dimensional complex
space.
However, our system naturally comes with a single parameter $R$, which we have seen in Section \ref{sec:prelims} together with the almost Mathieu operator. 

Heuristically, this is the reason why the ``correct'' representation for our purposes is $Z_\lambda=Z_{0,0}(\lambda)$. A more formal
argument goes like this: the representation in Eq. \eqref{eq:artrep} has nonzero entries only along the
secondary diagonals
with zero corner elements.
The parameters $a,b$ do appear in the corners of $E$ and $F$, thus we need to
set them to zero. The following theorem illustrates the situation precisely.

\begin{theorem}[Ref.~\onlinecite{jan:qtumgrp}]
Let $q^2$ be a primitive $\ell$-th root of unity. The irreducible finite-dimensional representations of $\U(\qsl)$ are
organized in three families, denoted by $L(n,+)$, $L(n,-)$, and $Z_{a,b}(\lambda)$. The representations $L(n,\pm)$ have dimension
smaller than or equal to $\ell-1$. The canonical matrix forms of $E$ and $F$ under $Z_{a,b}(\lambda)$ have zero corner entries if
and only if $a=b=0$. In this case, the matrix forms are:
\begin{align*}
 Z_{\lambda}(E) &= -\begin{pmatrix}
 0 & b_1 & & & \\
 & 0 & \ddots & \\
 & & \ddots & b_{Q-1}&\\
 & & & 0 
 \end{pmatrix}, &
 Z_{\lambda}(F) &= \begin{pmatrix}
 0 & & & \\
 1 & 0 & & \\
 & \ddots & \ddots & \\
 & & 1 & 0 
 \end{pmatrix},
\end{align*}
where we have set $Z_\lambda= Z_{0,0}(\lambda)$ and $b_r=(q-q\inv)^{-2}(q^r-q^{-r})(q^{r-1}\lambda\inv-q^{1-r}\lambda)$, and the values $\lambda=\pm 1,\pm q, \dots, \pm q^{Q-2}$
are \emph{not} allowed.
\end{theorem}

Set $\zeta=Z_{q\inv}$ and $\Hmatrix^\prime=\ii(q-q\inv)\zeta(F-E)$. We have the following proposition.
\begin{proposition}\label{prop:poltheta}
 The characteristic polynomial of $\Hmatrix$ equals that of $\Hmatrix^\prime$,
 \begin{gather*}
 f(E)=\det(\Hmatrix-E)=\det(\Hmatrix^\prime-E),\\
 \Hmatrix^\prime= -2\sin(\pi\alpha)\begin{pmatrix}
 0 & \frac{\sin^2(\pi\cdot 1 \cdot\alpha)}{\sin^2(\pi\alpha)} & & \\
 1 & 0 & \ddots & \\
 & \ddots & \ddots & \frac{\sin^2(\pi(Q-1)\alpha)}{\sin^2(\pi\alpha)}\\
 & & 1 & 0
 \end{pmatrix}.
 \end{gather*} 
\end{proposition}
\begin{proof}
 By gauge invariance of $H_{\text{Hof}}$, we can compute $\det(\Hmatrix-E)$ when the vector potential is chosen as
 $\vec{A}= \alpha/2\cdot (-x-y,\, x+y+1)$. 
In this case, the results in Ref.~\onlinecite{wiegmann_bethe-ansatz_1994} imply that $\Hmatrix$ is given as prescribed by Eq. \eqref{eq:magrep}
 with $\varphi=\zeta_0$.
According to
 the classification 
 of the quantum group representations, $\zeta_0$ must be equivalent to $Z_\lambda$ for
 some $\lambda$. Indeed, a direct computation (by conjugation with a diagonal matrix) shows that $\lambda=q\inv$ (note that,
 if $\Lambda$ is such diagonal matrix, the system resulting from $\Lambda\zeta_0(F)\Lambda\inv=
 \zeta(F)$ is overdetermined). 
 Equivalent representations are related through similar matrices, 
i.e., which are equivalent up to a similarity (but not necessarily unitary) transformation,
 and it is well-known that 
 similarity preserves the characteristic polynomial.
\end{proof}

\begin{remark}\label{rem:nonhermitian}
Note that $\Hmatrix$ is hermitian, whereas $\Hmatrix^\prime$ is not.
However, the matrices $\Hmatrix$ and $\Hmatrix^\prime$ are equivalent up to a similarity transformation, and therefore isospectral with real spectra:
Indeed, $\Hmatrix^\prime$ is pseudo-hermitian, following the definitions introduced by Mostafazadeh\cite{mostafazadeh_pseudo-hermiticity_2002}.
Hence, the corresponding Hamiltonian $H^\prime$ is an example of a non-Hermitian Hamiltonian with real spectrum which is equivalent up to a similarity transformation to a more conventional hermitian Hamiltonian\cite{mostafazadeh_pseudo-hermiticity_2002,fernandez_non-hermitian_2015}.
\end{remark}

\begin{remark}\label{rem:cheb}
We see from Definition \ref{def:hoffindim} that (at least for $\nu_y=0$) the diagonal of $\Hmatrix$ is given by the Chebyshev polynomials of the first kind, in symbols $\theta_k=2T_{k-1}(\cos(2\pi\alpha))$. On the other hand, the Hamiltonian $\Hmatrix^\prime$ in the new ``quantum group-adjusted'' gauge\cite{wiegmann_bethe-ansatz_1994} shows the Chebyshev polynomials of the \emph{second} kind on the upper diagonal, that is, $\Hmatrix^\prime_{k,k+1}=U_{k-1}^2(\sin(\pi\alpha))$.
\end{remark}

\begin{remark}\label{rem:hopcost}
It is possible to take into account the 
anisotropy parameter
$R$ (i.e., the parameter which multiplies the cosine function in the almost Mathieu operator) in the expression for $\Hmatrix^\prime$ by choosing the representation $Z_{Rq\inv}$. The Hamiltonian will then be written as $Z_{Rq\inv}(RF - E)$. 
It is interesting to note that the parameter controlling the representation of the quantum group coincides with the parameter controlling the anisotropy of the system.
\end{remark}

The advantage of $\zeta$ over $\zeta_0$ is twofold: the freedom in choosing the parameter allows us to write an expression that is independent of the parity of $P$, and $\zeta(F)$ is represented through a matrix with a constant (secondary) diagonal, which makes
the formula for $\det(\Hmatrix^\prime-E)$ more easily guessed.

Recall that the \emph{elementary symmetric polynomials} appear when we expand a linear factorization of a monic polynomial:
\begin{align*}
\prod_{i=1}^n ( \lambda - x_i)=&\,\,\lambda^n - e_1(x_1,\ldots,x_n)\lambda^{n-1} + e_2(x_1,\ldots,x_n)\lambda^{n-2} + \cdots +\\
&(-1)^n e_n(x_1,\ldots,x_n),\\
e_k (x_1 , \ldots , x_n )=&\,\,\sum_{1\le j_1 < j_2 < \cdots < j_k \le n} x_{j_1} \dotsm x_{j_k}.
\end{align*} 
We shall need the following ``$2$-step modification'' of the $e_k$'s:
\begin{equation*}
\mode_k (x_1 , \ldots , x_n )=\,\,\sum_{\substack{1\le j_1 < j_2 < \cdots < j_k \le n\\
\lvert j_i-j_{i+1}\rvert \geq 2}} x_{j_1} \dotsm x_{j_k}.
\end{equation*} 
We also use the convention $\mode_0=1$. 

The following lemma is easily verified.
\begin{lemma}\label{lem:combsym}
Suppose $N$ is odd and $k=1,\dots,\lfloor N/2\rfloor$. We have the identity
\begin{equation*}
\mode_k(x_1,\dots,x_{N-1})=\mode_k(x_1,\dots,x_{N-2})+x_{N-1}\mode_{k-1}(x_1,\dots,x_{N-3}).
\end{equation*}
Suppose $N$ is even and $k=1,\dots,N/2-1$. We have the identities
\begin{align*}
\mode_k(x_1,\dots,x_{N-1}) &=\mode_k(x_1,\dots,x_{N-2})+x_{N-1}\mode_{k-1}(x_1,\dots,x_{N-3}),\\
\mode_{N/2}(x_1,\dots,x_{N-1}) &=x_{N-1}\mode_{N/2-1}(x_1,\dots,x_{N-3}).
\end{align*}
\end{lemma}

Let us introduce a simple formula for the determinant of tridiagonal matrices.

\begin{theorem}\label{thm:formula}
 Let $A=A(N)$ be the $N$-by-$N$ matrix with complex-valued entries described below. Then the determinant of $A$ is given by
 \begin{align*}
 A&=\begin{pmatrix}
 x & b_1 & & \\
 y & x & \ddots & \\
 & \ddots & \ddots & b_{N-1}\\
 & & y & x 
 \end{pmatrix}\\
 \det(A)&=\sum_{i=0}^{\lfloor N/2\rfloor} (-1)^i x^{N-2i}\,y^i \,\mode_i(b_1,b_2,\dots,b_{N-1}).
 \end{align*}
\end{theorem}
\begin{proof}
 We proceed by induction on $N$. 
 Cases $N=1,2$ are straightforward. 
 For the induction step, notice that the standard Laplace expansion
 of the determinant (starting from the bottom right) yields the recurrence relation
 \[ \det A(N)=x\cdot \det A(N-1)-yb_{N-1}\cdot \det A(N-2).\]
 Suppose $N$ is odd. Notice that $(N-1)/2=\lfloor N/2\rfloor$ and $\lfloor (N-2)/2\rfloor=\lfloor N/2\rfloor-1$. We can compute the summation as follows:
 \begin{align*}
	x\cdot \det A(N-1)&=\sum_{i=0}^{\lfloor N/2\rfloor} (-1)^i x^{N-2i}\,y^i \,\mode_i(b_1,b_2,\dots,b_{N-2})\\
	-yb_{N-1}\cdot \det A(N-2)&=\sum_{i=0}^{\lfloor N/2\rfloor-1} (-1)^{i+1} x^{N-2(i+1)}\,y^{i+1}b_{N-1} \,\mode_i(b_1,b_2,\dots,b_{N-3})\\
	&=\sum_{i=1}^{\lfloor N/2\rfloor} (-1)^{i} x^{N-2i}\,y^{i}b_{N-1} \,\mode_{i-1}(b_1,b_2,\dots,b_{N-3})
 \end{align*}
 where in the last equality we reindexed starting from $1$ rather than $0$. When $i=0$ we get $x^N$
as expected. 
For the other indices we can sum the two summations and use Lemma \ref{lem:combsym} to obtain the result. 
Analogously, the case where $N$ is even is carried out 
using the corresponding identities in Lemma \ref{lem:combsym}.
\end{proof}

To calculate the determinants of tridiagonal matrixes with corners, we will need the following result.

\begin{corollary}\label{cor:corncont}
 Let $A^b,A_b$ be the $N$-by-$N$ matrices with complex-valued entries described below. Their determinants are given by
 \begin{align*}
 A^b&=\begin{pmatrix}
 x & b_1 & & b \\
 y_1 & x & \ddots & \\
 & \ddots & \ddots & b_{N-1}\\
 & & y_{N-1} & x 
 \end{pmatrix}&
 A_b&=\begin{pmatrix}
 x & b_1 & & \\
 y_1 & x & \ddots & \\
 & \ddots & \ddots & b_{N-1}\\
 b & & y_{N-1} & x 
 \end{pmatrix}
 \end{align*} 
 \begin{align*}
 \det(A^b)&=\det(A)+(-1)^{N-1}b\cdot y_1\cdots y_{N-1},\\
 \det(A_b)&=\det(A)+(-1)^{N-1}b\cdot b_1b_2\cdots b_{N-1}.
 \end{align*}
\end{corollary}
\begin{proof}
 The determinants can be computed using the matrix determinant lemma: given row vectors $u,v$ of length $N$, the formula
 $\det(A+u^tv)=\det(A)+v\cdot \mathrm{adj}(A)u^t$ holds (the subscript $t$ indicates transposition). 
 Recall that 
 the adjugate matrix $\mathrm{adj}(A)$ is the matrix such that $A\cdot \mathrm{adj}(A)=\det(A)$, and it is computed
 by taking the transpose of the cofactor matrix of $A$ (the matrix of signed minors).

 By setting $u=(1,0,\dots,0)$ and $v=(1,0,\dots,0)$ we can write $A^b=A+bu^tv$. Since $v\cdot\mathrm{adj}(A)u^t=
 \mathrm{adj}(A)_{n,1}$, we need to compute the $(1,n)$-minor of $A$. The corresponding submatrix is upper triangular,
 with the diagonal equal to the lower secondary diagonal of $A$. This implies $v\cdot\mathrm{adj}(A)u^t=y_1\cdots y_{N-1}$. 
 We obtain $\det(A^b)=\det(A)+(-1)^{N-1}b\cdot y_1\cdots y_{N-1}$.
The formula for $A_b$ can be proven analogously. 
\end{proof}

Combining the previous theorem with Proposition \ref{prop:poltheta}, we get the main result on the polynomial $f$. 
Recall that $\alpha=P/Q$ is a rational number in reduced form.
\begin{theorem}\label{thm:formula_conc}
We have a formula for $f(E)=\det(\Hmatrix-E)$ given as
\begin{equation*}
f(E)=\sum_{i=0}^{\lfloor Q/2\rfloor} (-1)^{Q+i}\,4^i\,E^{Q-2i}\,
\mode_i(\sin^2(\pi\alpha),\sin^2(2\pi\alpha),\dots,\sin^2((Q-1)\pi\alpha) ).
\end{equation*}
\end{theorem}

\begin{remark}\label{rem:frem}
Let us note that the formula from Theorem \ref{thm:formula} makes at least two aspects clear: if we multiply $y$ by a quantity $x$ and the $b_i$'s by the inverse $x\inv$, the determinant of $A$ is unaffected. Moreover, when $N$ is odd, the determinant goes to zero if we set $x=0$. 
We thus recover the well-known result mentioned earlier:
the spectrum of the Hofstadter model at the mid-band point always contains $0$ when $Q$ is odd.
\end{remark}

\begin{remark}
As previously noted, since $f(0)= f( 0,{\pi}/{2Q},{\pi}/{2Q})=4(-1)^{Q/2}$ if $Q$ is even,
the equation above yields the identity
\begin{equation}
\mode_{Q/2}(\sin^2(\pi\alpha),\sin^2(2\pi\alpha),\dots,\sin^2((Q-1)\pi\alpha))
=4^{-(Q/2-1)}
,
\end{equation}
for even $Q$. 
\end{remark}

\section{Chambers relation and Chebyshev polynomials}

We are left with the question of determining the general dependence of the energy on the Brillouin torus. We want to emphasize that guessing the general form of $\Hmatrix^\prime$ is facilitated by drawing from all the relationships we have established so far with the Toda model, the quantum group, and the Chebyshev polynomials.

Firstly, we know from Section \ref{sec:prelims} that these Hamiltonians have the form $(q-q\inv)\phi(E+F)$ for a suitable representation. 
Secondly, we obtained an expression of the form $\Hmatrix^\prime=\ii(q-q\inv)\zeta(-E+F)=(q-q\inv)\zeta(\ii\inv E+\ii F)$. 
The natural guess away from the mid-band point is to set 
\begin{equation}\label{eq:temph}
\Hmatrix^\prime(\kappa_+,\kappa_-)=(q-q\inv)\zeta(\ee^{-\ii\kappa_-} E+\ee^{\ii\kappa_-}F),
\end{equation}
where the variables $\kappa_\pm$ are parametrizing the Brillouin zone. The dependence on $\kappa_+$ can be introduced by reverting the Chebyshev polynomials from the second to the first kind. In other words, we can modify $\zeta_0$ from Eq. \eqref{eq:artrep} to
\begin{equation*}
 \zeta_c(E)_{k+1,k}=\frac{\ee^{\ii\kappa_+}q^{k}+\ee^{-\ii\kappa_+}q^{-k}}{q-q\inv},\qquad \zeta_c(F)_{k,k+1}=\frac{\ee^{\ii\kappa_+}q^{k}+\ee^{-\ii\kappa_+}q^{-k}}{q-q\inv},
\end{equation*}
where we note that $\ee^{\ii\kappa_+}q^{k}+\ee^{-\ii\kappa_+}q^{-k}=2\cos(k\pi\alpha+\kappa_+)$.

Lastly, since the general form of the representation $Z_{a,b}(\lambda)$ includes corners (for nontrivial $a$ and $b$), we can interpret $\zeta_c$ above periodically (continuing the diagonals from the top when they reach the bottom), which has the effect of filling up precisely the two corners. Overall, we obtain the following form:
\begin{gather*}
 \Hmatrix^\prime_c(\kappa_+,\kappa_-)=(q-q\inv)\zeta_c(z\inv E+zF)=\\
 \begin{pmatrix}
 0 & \begin{smallmatrix}z \cdot 2\cos(1\cdot\pi\alpha+\kappa_+)\end{smallmatrix} &
 & \begin{smallmatrix}z\inv\cdot 2\cos(Q\cdot\pi\alpha+\kappa_+)\end{smallmatrix} \\
 \begin{smallmatrix}z\inv\cdot 2\cos(1\cdot\pi\alpha+\kappa_+)\end{smallmatrix} & 0 & \ddots & \\
 & \ddots & \ddots & \begin{smallmatrix}z \cdot 2\cos((Q-1)\cdot\pi\alpha+\kappa_+)\end{smallmatrix}\\
 \begin{smallmatrix}z \cdot 2\cos(Q\cdot\pi\alpha+\kappa_+)\end{smallmatrix} & 
 & \begin{smallmatrix}z\inv\cdot 2\cos((Q-1)\cdot\pi\alpha+\kappa_+)\end{smallmatrix} & 0 
 \end{pmatrix},
\end{gather*}
where we have set $z=\ee^{\ii(\kappa_- + \frac{\pi}{2})}$. 
Note that, compared to Eq. \eqref{eq:temph}, we shifted $\kappa_-\mapsto \kappa_- + \frac{\pi}{2}$, so that when $(\kappa_+,\kappa_-)=(\pi/2,0)$ we are in a situation equivalent to the mid-band point representation (note that the corners of $\Hmatrix^\prime_c$ vanish). For ease of comparison, we will also change Eq. \eqref{eq:chambers} and center the mid-band point at the origin by setting $\vec{\nu}=\vec{\nu}^\prime + (\pi/2Q,\pi/2Q)$, so that the second term in the new coordinates reads
\begin{equation}\label{eq:nch}
(-1)^{Q}\cdot 2(\sin(Q\nu_x^\prime)+\sin(Q\nu_y^\prime)).
\end{equation}

We can now use Corollary \ref{cor:corncont} to quantify the corner contributions, leading us to the Chambers relation in the coordinates $\kappa_\pm$. 
Let us consider the contribution of one corner, multiplied by $z^Q+z^{-Q}$, 
\begin{equation}\label{eq:cham}
2(-1)^{Q-1}\cos(Q\Bigl(\kappa_- + \frac{\pi}{2}\Bigr))\prod_{i=1}^Q 2\cos(i\cdot\pi\alpha+\kappa_+).
\end{equation}
We claim that Eq. \eqref{eq:cham} prescribes the energy dependence on $\vec{\nu}$ (see Theorem \ref{thm:par} and its proof for the precise statement).

To prove the claim, inspired by Remark \ref{rem:cheb}, we could proceed by induction on $Q$ and use the recurrence relations of the Chebyshev polynomials in the proof. 
Although this seems viable, we will follow a more direct approach.

\begin{lemma}\label{lem:mfor}
The following identity holds:
\begin{align*}
\prod_{j=1}^Q 2\cos(\frac{j \pi P}{Q} +\kappa)&=\begin{cases} 
\exp(\ii \frac{\pi}{2}[P(Q+1)+1]) \cdot 2 \sin(Q\kappa)& \text{if $Q$ is even}\\
\exp(\ii \frac{\pi}{2}P(Q+1)) \cdot 2 \cos(Q\kappa)& \text{if $Q$ is odd}
\end{cases}\\
&=\exp(\ii \frac{\pi}{2}(P+1)(Q+1))\sin(Q(\kappa+\pi/2)).
\end{align*}
\end{lemma}
Note that $P$ cannot be even when $Q$ is even, hence the exponential factor above is always equal to either $1$ or $-1$.
\begin{proof}
We use the complex exponential formula for cosine:
\[
\prod_{j=1}^Q e^{\ii\bigl(\frac{ j \pi P}{Q} +\kappa\bigr)}+e^{-\ii\bigl(\frac{\ii j \pi P}{Q} +\kappa\bigr)}=\prod_{j=1}^Q\xi^{-j}e^{-\ii\kappa}(\xi^{2j}e^{2\ii\kappa}+1),
\]
where we set $\xi:=e^{\ii\frac{\pi P}{Q}}$. Note that $\xi^2$ is a primitive root of unity. We continue the computation:
\[
\xi^{-(1+\cdots+Q)}e^{-\ii Q\kappa}\prod_{j=1}^Q \xi^{2j} (e^{2\ii\kappa}- (-\xi^{-2j}))=\xi^{\frac{Q(Q+1)}{2}}e^{-\ii Q\kappa}\prod_{j=1}^Q (e^{2\ii\kappa}- (-\xi^{2j})),
\]
where in the last equality we used that roots of unity are symmetric with respect to the real axis (i.e., they are invariant under inversion). Considering the last product as a polynomial in the variable $e^{2\ii\kappa}$, we see that
\[
\prod_{j=1}^Q (e^{2\ii\kappa}- (-\xi^{2j}))=e^{2\ii Q\kappa} \pm 1,
\]
where the upper sign corresponds to the case where $Q$ is odd (i.e., $-\xi^2$ is a root of $-1$), and the lower sign corresponds to the case where $Q$ is even (in this case roots of unity are symmetric with respect to the origin, that is invariant under negation). Thus, we can conclude as follows:
\[
\xi^{\frac{Q(Q+1)}{2}}e^{-\ii Q\kappa}(e^{2\ii Q\kappa} \pm 1)=e^{\ii\frac{\pi}{2}P(Q+1)}(e^{\ii Q\kappa} \pm e^{-\ii Q\kappa}).
\]

\end{proof}

After this auxiliary result, we are ready to prove the last theorem of the paper.

\begin{theorem}\label{thm:par}
Consider the change of coordinates 
$\kappa_\pm=\frac12\epsilon(P,Q){((-1)^Q \nu^\prime_x\pm \nu^\prime_y)}$, 
where $\epsilon$ is a sign given as 
\[ 
\epsilon(P,Q)=\begin{cases}
(-1)^{ \frac{P-1}{2}} & \text{if $Q$ is even}\\
(-1)^{ \frac{(Q+1)(P+1)}{2}} & \text{if $Q$ is odd.}
\end{cases}
\]
Then the Chambers relation in the representation $\zeta_c$ is
\[
\det(\Hmatrix^\prime_c(\kappa_+,\kappa_-)-E)=\det(\Hmatrix^\prime_c(0,0)-E)+g(P,Q,\kappa_+,\kappa_-),
\]
where $g$ is a function determined as follows:
\[ 
g(P,Q,\kappa_+,\kappa_-)=\begin{cases}
\epsilon(P,Q) \cdot 4 \cos(Q\kappa_-)\sin(Q\kappa_+) & \text{if $Q$ is even}\\
\epsilon(P,Q) \cdot 4 \sin(Q\kappa_-)\cos(Q\kappa_+)  & \text{if $Q$ is odd.}
\end{cases}
\]
\end{theorem}
Note that the change of coordinates in the odd case equals the change of coordinates in the even case up to a sign and a swap of $\kappa_+$ with $\kappa_-$.

Before the proof, let us also note that the fundamental property of the Chambers relation, namely the relation 
\[
\det(\Hmatrix^\prime_c(x)-z)=\det \Hmatrix^\prime_c(x_0)-g(x_0)+g(x),\]
can be proved by standard techniques, see for example Ref.~\onlinecite[Section 4]{kreft_94}. We will not go into the details here and rather focus on the particular form of the ``offset'' function $g$.

\begin{proof}
From Eq. \eqref{eq:chambers} and Proposition \ref{prop:poltheta} we know we only need to focus on the function $g$. From Eq. \eqref{eq:cham} and Lemma \ref{lem:mfor}, we obtain the corner contribution as follows:
\[ 
(-1)^{Q-1}\exp(\ii \frac{\pi}{2}(P+1)(Q+1)) \cdot\begin{cases} 
4 \cos(Q\kappa_-)\sin(Q\kappa_+)& \text{if $Q$ is even}\\
4 \sin(Q\kappa_-)\cos(Q\kappa_+)& \text{if $Q$ is odd.}
\end{cases}
\]
If $Q$ is even, $P$ must be odd, and the factor before the curly bracket above depends on whether $P+1$ is singly or doubly even. Since $(-1)^{Q-1}$ is negative, the singly even case corresponds to a positive sign, while the doubly even case gives a negative sign. 
Overall, the sign is $(-1)^{\frac{P-1}{2}}$. 
When $Q$ is odd, the part preceding the curly bracket above is positive if $Q$ is congruent to $3$ modulo $4$, i.e., if $\lceil \frac{Q}{2}\rceil$ is even. 
In the other case, the sign depends on the parity of $P+1$. 
Overall, we see the sign can be expressed as $(-1)^{ \frac{(Q+1)(P+1)}{2}}$. 
It is left to verify that $g$ corresponds to Eq. \eqref{eq:nch}. Under the change of variables in the statement, we have that
\[
\nu^\prime_x=(-1)^Q\epsilon (\kappa_+ +\kappa_-),\quad \nu^\prime_y=(-1)^Q\epsilon [(-1)^Q(\kappa_+ - \kappa_-)].
\]
Since $\sin(-x)=-\sin(x)$, and using standard trigonometric formulas, Eq. \eqref{eq:nch} gets transformed to
\[
2\epsilon[\sin(\kappa_+ +\kappa_-)+\sin((-1)^Q(\kappa_+ - \kappa_-))]=\begin{cases}
4\epsilon \cos(Q\kappa_-)\sin(\kappa_+) & \text{if $Q$ is even}\\
4\epsilon \sin(Q\kappa_-)\cos(\kappa_+)  & \text{if $Q$ is odd.}
\end{cases}
\]
\end{proof}

Let us end this section with a small remark about the gauge associated with $\Hmatrix^\prime$ (similar considerations hold for $\Hmatrix_c^\prime$). As noted in Ref.~\onlinecite{wiegmann_bethe-ansatz_1994}, the Harper equation at the mid-band point in this setting becomes a difference functional equation:
\[
\ii(z\inv+qz)\Psi(qz)-\ii(z\inv+q\inv z)\Psi(q\inv z)=E\Psi(z).
\]
Recall that the form of the vector potential is chosen to be $\vec{A}= \alpha/2\cdot (-x-y,\, x+y+1)$. 
Hence a convenient coordinate for the Bloch wave function is $l=n+m$. The Harper equation for $\psi=(\psi_l)$ is obtained from the equation above by setting $z=q^l$ and $\psi_l=\Psi(q^l)$. Notice that $l$ ranges in $1,\dots, 2Q$, presenting a ``doubling'' of the period in comparison with the Harper equation in Section \ref{sec:prelims}. Clearly, this is also reflected in the entries of $\Hmatrix^\prime$ where $\pi\alpha$ (rather than $2\pi\alpha$) appears.

On this basis, it is arguable that the representation $\zeta_c$ should be extended to matrices of order $2Q$ and that the characteristic polynomial of $\Hmatrix^\prime$ should be computed under this convention. As the sine and cosine functions are $\pi$-periodic up to a sign, it is not hard to prove that the resulting characteristic polynomial is the square of the one computed above. 

More precisely, using Theorems \ref{thm:formula} \& \ref{thm:formula_conc}, we have that $b_j=\sin^2(j\pi\alpha)/\sin^2(\pi\alpha)$, so that $b_Q=0$, and $b_{Q+j}=b_j$. Since the variables $b_j$ enter the formula only through 
polynomials which are a
modification of the elementary symmetric polynomials, these are computed by expanding the linear factorization of the monic polynomial $\lambda^2 \prod_1^{Q-1} (\lambda-b_i)^2$.

From here it can be seen that $\det(\Hmatrix^\prime_{2Q}-E)=\left(\det(\Hmatrix^\prime_{Q}-E)\right)^2$, from which we deduce that the period doubling in the Schrödinger equation above does not affect the eigenvalues up to multiplicity.

\section*{Conclusions}

Concluding, we discussed and clarified the spectral relationship between the Hofstadter model in condensed matter physics and the relativistic Toda lattice in high-energy physics found by Hatsuda, Katsura, and Tachikawa\cite{hatsuda_hofstadters_2016} in the framework of the representation theory of the elementary quantum group.
Furthermore, we derived a formula parametrizing the energy spectrum of the Hofstadter model in the Brillouin zone in terms of elementary symmetric polynomials and Chebyshev polynomials, building on previous work on 
the Hofstadter model by Wiegmann and Zabrodin\cite{wiegmann_bethe-ansatz_1994} 
and on
tridiagonal matrices by Molinari\cite{mol:tri}.
We hope that our work will serve as a basis for a deeper understanding of 
the self-similarity properties
of the Hofstadter butterfly and contribute to shed light on the connection between the Hofstadter and the relativistic Toda lattice models and, more generally, on the connection between noncommutative quantum geometry and the quantum world.

\subsection*{Acknowledgements}
P.~M. thanks Hosho Katsura for useful discussions.
P.~M. is supported by the Japan Science and Technology Agency (JST) of the Ministry of Education, Culture, Sports, Science and Technology (MEXT), CREST Grant~No.~JPMJCR19T2, the Japan Society for the Promotion of Science (JSPS) Grant-in-Aid for Early-Career Scientists KAKENHI Grants No.~JP23K13028 and No.~JP20K14375.
V.~P.~is supported by the JST CREST Grant.~No.~JPMJCR19T2 and by Marie Skłodowska-Curie Individual Fellowship (project number 101063362).
X.~S.~is partially supported by KAKENHI Grant No. JP21K03259 and Grants from Joint Project of OIST, Hikami Unit and the University of Tokyo.

\subsection*{Data Availability Statement}

Data sharing is not applicable to this article as no new data were created or analyzed in this study.


%

\end{document}